\documentclass[abstract,headings=normal,DIV=12]{scrartcl}

\usepackage[automark]{scrpage2}
\usepackage[english]{babel}
\usepackage[T1]{fontenc}
\usepackage[applemac]{inputenc}
\usepackage{graphicx}
\usepackage{tikz}
\usepackage{amsmath}
\usepackage{amssymb}
\usepackage{amsthm}
\usepackage{subfig}
\usepackage{accents}
\usepackage{url}
\usepackage[colorlinks=true,linkcolor=blue,citecolor=red,urlcolor=black]{hyperref}

\usetikzlibrary{arrows}

\addtokomafont{caption}{\small}

\newcommand{\C}{\mathbb{C}}

\newcommand{\Z}{\mathbb{Z}}

\newcommand{\Ell}{\mathcal{L}}

\newcommand{\X}{\mathcal{X}}

\DeclareMathOperator{\sn}{sn}

\renewcommand{\tilde}[1]{\widetilde{#1}}

\makeatletter
\renewcommand*\env@cases[1][1.2]{%
  \let\@ifnextchar\new@ifnextchar
  \left\lbrace
  \def\arraystretch{#1}%
  \array{@{}l@{\quad}l@{}}%
}
\makeatother

\newtheorem{theo}{Theorem}[section]
\newtheorem{lemma}[theo]{Lemma}
\newtheorem{prop}[theo]{Proposition}
\newtheorem{cor}[theo]{Corollary}

\theoremstyle{definition}
\newtheorem{defi}[theo]{Definition}
\theoremstyle{remark}
\newtheorem*{rem}{Remark}

\graphicspath{{./graphics/}}
\DeclareGraphicsExtensions{.pdf}

\begin{document}

\title{On integrability of discrete variational systems. Octahedron relations}
\author{Raphael~Boll\and Matteo~Petrera\and Yuri~B.~Suris}
\publishers{\vspace{0.5cm}{\small Institut f\"ur Mathematik, MA 7-2, Technische Universit\"at Berlin,\\
Stra{\ss}e des 17. Juni 136, 10623 Berlin, Germany\\
E-mail: \url{boll}, \url{petrera}, \url{suris@math.tu-berlin.de}}}
\maketitle

\begin{abstract}
\noindent We elucidate consistency of the so-called corner equations which are elementary building blocks of Euler-Lagrange equations for two-dimensional pluri-Lagrangian problems. We show that their consistency can be derived from the existence of two independent octahedron relations. We give explicit formulas for octahedron relations in terms of corner equations.

\medskip\noindent\emph{Subject areas:} mathematical physics, differential equations, field theory

\medskip\noindent\emph{Keywords:} discrete integrable systems, Euler-Lagrange equations, variational systems, pluri-Lagrangian systems, multi-dimensional consistency, fractional ideals
\end{abstract}

\section{Introduction}

This paper contributes to the Lagrangian theory of discrete integrable systems. This theory emerged recently, following several important developments. Its starting point is the understanding of integrability of discrete hyperbolic systems as their multi-dimensional consistency \cite{BS0, N}. This major breakthrough led to the classification of integrable quad-equations (discrete two-dimensional hyperbolic systems) \cite{ABS1}, which turned out to be rather influential. A further conceptual development was initiated by Lobb and Nijhoff \cite{LN} and deals with the variational (Lagrangian) formulation of discrete multi-dimensionally consistent systems. Their original idea can be summarized as follows: solutions of integrable quad-equations deliver critical points for actions along all possible quad-surfaces in multi-time; the Lagrangian form is closed on solutions.

Solutions of hyperbolic quad-equations do not exhaust critical points for actions along all possible quad-surfaces. In \cite{variational}, we  pushed forward the idea of considering the corresponding Euler-Lagrange equations and general solutions thereof as the proper notion of integrability for discrete variational systems. This is formalized in the following definition:
\begin{itemize}
\item Let $\Ell$ be a discrete 2-form, i.e., a real-valued function of oriented elementary squares
\[
\sigma_{ij}=(n,n+e_{i},n+e_{i}+e_{j},n+e_{j})
\]
of $\Z^m$, such that $\Ell(\sigma_{ij})=-\Ell(\sigma_{ji})$. It is assumed to depend on some field $x:\Z^m\to \X$  assigned to the points of $\Z^m$  ($\X$ being some vector space).
\item To an arbitrary oriented quad-surface $\Sigma$ in $\Z^m$, there corresponds the \emph{action functional}, which assigns to $x|_{V(\Sigma)}$, i.e., to the fields at the vertices of the surface $\Sigma$, the number
\begin{equation*}
S_{\Sigma}:=\sum_{\sigma_{ij}\in\Sigma}\Ell(\sigma_{ij}).
\end{equation*}
\item We say that the field $x:V(\Sigma)\to \X$ is a critical point of $S_{\Sigma}$, if at any interior point $n\in V(\Sigma)$, we have
\begin{equation}\label{eq: dEL gen}
    \frac{\partial S_{\Sigma}}{\partial x(n)}=0.
\end{equation}
Equations \eqref{eq: dEL gen} are called \emph{discrete Euler-Lagrange equations} for the action $S_{\Sigma}$.
\item We say that the field $x:\Z^m\to\X$ solves the \emph{pluri-Lagrangian problem} for the Lagrangian 2-form $\Ell$ if, \emph{for any quad-surface $\Sigma$ in $\Z^m$}, the restriction $x|_{V(\Sigma)}$ is a critical point of the corresponding action $S_{\Sigma}$.
\end{itemize}

Discrete Euler-Lagrange equations for the surface $\Sigma$ of a unit lattice cube are called {\em corner equations} and are, so to say, elementary particles, of which all possible Euler-Lagrange equations are built of. In general, the system of corner equations for one unit cube is heavily overdetermined, and its consistency is, in our view, synonymous with the integrability of the corresponding pluri-Lagrangian problem. We refer to \cite{variational} for details, as well as for some bibliographical and historical remarks concerning this definition.

It is the purpose of the present paper to contribute to a better understanding of algebraic mechanisms behind consistency of the system of corner equations for a class of discrete 2-forms coming from quad-equations of the ABS-list \cite{LN,BS1}. The corresponding system of corner equations consists of six equations per elementary cube $(x,x_{1},x_{2},x_{3},x_{12},x_{23},x_{13},x_{123})$, each depending on five out of the six variables $x_{1}$, $x_{2}$, $x_{3}$, $x_{12}$, $x_{23}$ and $x_{13}$. The system is \emph{consistent} if it has minimal possible rank $2$, i.e., if exactly two of these equations are independent. We will demonstrate that one can view consistency of the corner equations as a corollary of the existence of two \emph{octahedron relations}. The latter are multi-affine relations for the six variables $x_{1}$, $x_{2}$, $x_{3}$, $x_{12}$, $x_{23}$ and $x_{13}$, satisfied on each solution of corner equations, and, in their turn, having all six corner equations as their corollaries.


\section{Consistent systems of corner equations}\label{sec:corner}
\begin{defi}[System of corner equations]\label{def:corner system}
For a given discrete 2-form $\Ell$, the \emph{system of corner equations} consists of the discrete Euler-Lagrange equations for a surface of an elementary 3D cube.
\end{defi}
In the present paper, we consider consistent systems of corner equations having their origin in 3D consistent systems of quad-equations from the ABS-list \cite{ABS1}. Recall \cite{ABS1, LN} that the corresponding discrete 2-forms are of the following special, three-point shape:
\begin{equation*}
\Ell(\sigma_{ij})=\Ell(X,X_{i},X_{j};\alpha_{i},\alpha_{j})=
L(X,X_{i};\alpha_{i})-L(X,X_{j};\alpha_{j})-\Lambda(X_{i},X_{j};\alpha_{i},\alpha_{j}).
\end{equation*}
Setting
\[
\frac{\partial L(X,X_{i};\alpha_{i})}{\partial X}=\psi(X,X_{i};\alpha_{i}),\quad
\frac{\partial \Lambda(X_{i},X_{j};\alpha_{i},\alpha_{j})}{\partial X_{i}}=\phi(X_{i},X_{j};\alpha_{i},\alpha_{j}),
\]
we arrive at the following four-leg corner equations for the vertices $x_i$ and $x_{ij}$, respectively:
\begin{align}
  &  \psi(X_i,X_{ij};\alpha_j)-\psi(X_i,X_{ik};\alpha_k)-
    \phi(X_{i},X_{k};\alpha_i,\alpha_k)+\phi(X_{i},X_{j};\alpha_i,\alpha_j)=0, \label{eq: corner i} \tag{$E_i$}\\
  &  \psi(X_{ij},X_{i};\alpha_j)-\psi(X_{ij},X_{j};\alpha_i)-
    \phi(X_{ij},X_{ik};\alpha_j,\alpha_k)+\phi(X_{ij},X_{jk};\alpha_i,\alpha_k)=0. \label{eq: corner ij} \tag{$E_{ij}$}
\end{align}
Observe that the corner equations for the vertices $x$ and $x_{123}$ are vacuous, and that, moreover, corner equations \eqref{eq: corner i} and \eqref{eq: corner ij} do not involve the fields $X$ and $X_{123}$.

We recall the relation with the quad-equations: every multi-affine quad-equation from the ABS list,
\[
Q(x,x_i,x_j,x_{ij};\alpha_i,\alpha_j)=0,
\]
admits, after a certain change of variables $x=f(X)$, four equivalent three-leg forms, centered at each of the vertices of an elementary square $\sigma_{ij}$, for instance, the one centered at $x_i$ reads:
\[
Q=0\quad\Leftrightarrow\quad\psi(X_i,X;\alpha_i)-\psi(X_i,X_{ij};\alpha_j)-\phi(X_i,X_j;\alpha_i,\alpha_j)=0.
\]
Now, each of the four-leg corner equations is a sum of two three-leg equations for two adjacent elementary squares, e.g.,
equation \eqref{eq: corner i} is the difference of two three-leg forms centered at $x_i$:
\begin{align}
&\psi(X_i,X;\alpha_i)-\psi(X_i,X_{ij};\alpha_j)-\phi(X_i,X_j;\alpha_i,\alpha_j)=0, \label{eq: aux1}\\
&\psi(X_i,X;\alpha_i)-\psi(X_i,X_{ik};\alpha_k)-\phi(X_i,X_k;\alpha_i,\alpha_k)=0. \label{eq: aux2}
\end{align}

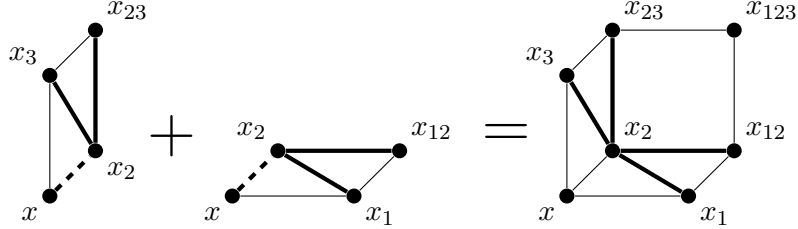
\begin{figure}[htbp]
   \centering
   \begin{tikzpicture}[auto,scale=0.4,inner sep=2]
      \node (x) at (0,0) [circle,fill,label=-135:$x$] {};
      \node (x2) at (1.5,1.5) [circle,fill,label=-45:$x_{2}$] {};
      \node (x3) at (0,4) [circle,fill,label=135:$x_{3}$] {};
      \node (x23) at (1.5,5.5) [circle,fill,label=45:$x_{23}$] {};
      \draw (x23) to (x3) to (x);
      \draw [ultra thick,dashed] (x) to (x2);
      \draw [ultra thick] (x2) to (x23);
      \draw [ultra thick] (x2) to (x3);
      \node (plus) at (4,2) {\Huge$+$};
      \node (x) at (6,0) [circle,fill,label=-135:$x$] {};
      \node (x1) at (10,0) [circle,fill,label=-45:$x_{1}$] {};
      \node (x2) at (7.5,1.5) [circle,fill,label=135:$x_{2}$] {};
      \node (x12) at (11.5,1.5) [circle,fill,label=45:$x_{12}$] {};
      \draw (x) to (x1) to (x12);
      \draw [ultra thick,dashed] (x) to (x2);
      \draw [ultra thick] (x2) to (x12);
      \draw [ultra thick] (x1) to (x2);
      \node (gleich) at (15,2) {\Huge$=$};
      \node (x) at (17,0) [circle,fill,label=-135:$x$] {};
      \node (x1) at (21,0) [circle,fill,label=-45:$x_{1}$] {};
      \node (x2) at (18.5,1.5) [circle,fill,label=45:$x_{2}$] {};
      \node (x3) at (17,4) [circle,fill,label=135:$x_{3}$] {};
      \node (x12) at (22.5,1.5) [circle,fill,label=45:$x_{12}$] {};
      \node (x23) at (18.5,5.5) [circle,fill,label=45:$x_{23}$] {};
      \node (x123) at (22.5,5.5) [circle,fill,label=45:$x_{123}$] {};
      \draw (x) to (x1) to (x12);
      \draw (x23) to (x3) to (x);
      \draw [ultra thick] (x23) to (x2) to (x12);
      \draw (x2) to (x);
      \draw [ultra thick] (x1) to (x2) to (x3);
      \draw (x12) to (x123) to (x23);
   \end{tikzpicture}
   \caption{A corner equation for a discrete three-point 2-form can be written as a sum of the three-leg forms of two adjacent quad-equations}
   \label{fig:four-leg}
\end{figure}
Thus, to every system of corner equations \eqref{eq: corner i}, \eqref{eq: corner ij} there corresponds a system of quad-equations:
\begin{equation}\label{eq:quadsystem}
\begin{aligned}
&Q_{12}:=Q(x,x_{1},x_{2},x_{12};\alpha_{1},\alpha_{2})=0,\qquad&
&\bar{Q}_{12}:=Q(x_{3},x_{13},x_{23},x_{123};\alpha_{1},\alpha_{2})=0,\\
&Q_{23}:=Q(x,x_{2},x_{3},x_{23};\alpha_{2},\alpha_{3})=0,&
&\bar{Q}_{23}:=Q(x_{1},x_{12},x_{13},x_{123};\alpha_{2},\alpha_{3})=0,\\
&Q_{13}:=Q(x,x_{3},x_{1},x_{13};\alpha_{3},\alpha_{1})=0,&
&\bar{Q}_{13}:=Q(x_{2},x_{23},x_{13},x_{123};\alpha_{3},\alpha_{1})=0,
\end{aligned}
\end{equation}
with the following property: every solution of the latter system satisfies the system of corner equations, but not vice versa.

The ABS list of irreducible multi-affine polynomials $Q\in\C[x,x_{1},x_{2},x_{12}]$ depending on two parameters $\alpha_{1}$, $\alpha_{2}$ is given in Appendix~\ref{sec:quadeqs}. It is known that for every system of quad-equations~\eqref{eq:quadsystem} from the ABS list, there exist two multi-affine quad-equations
\[
T(x,x_{12},x_{23},x_{13})=0\qquad\text{and}\qquad \bar{T}(x_{1},x_{2},x_{3},x_{123})=0
\]
satisfied for every solution of \eqref{eq:quadsystem} (the so called \emph{tetrahedron equations}).

\begin{theo}
The system of corner equations \eqref{eq: corner i}, \eqref{eq: corner ij} is equivalent to the following system of six polynomial equations:
\begin{equation}\label{eq:cornersystem}
\begin{aligned}
&E_{1}:=\frac{\partial Q_{12}}{\partial x}Q_{13}-\frac{\partial Q_{13}}{\partial x}Q_{12}=0,\qquad&
&E_{23}:=\frac{\partial\bar{Q}_{13}}{\partial x_{123}}\bar{Q}_{12}-\frac{\partial\bar{Q}_{12}}{\partial x_{123}}\bar{Q}_{13}=0,\\
&E_{2}:=\frac{\partial Q_{23}}{\partial x}Q_{12}-\frac{\partial Q_{12}}{\partial x}Q_{23}=0,&
&E_{13}:=\frac{\partial\bar{Q}_{12}}{\partial x_{123}}\bar{Q}_{23}-\frac{\partial\bar{Q}_{23}}{\partial x_{123}}\bar{Q}_{12}=0,\\
&E_{3}:=\frac{\partial Q_{13}}{\partial x}Q_{23}-\frac{\partial Q_{23}}{\partial x}Q_{13}=0,&
&E_{12}:=\frac{\partial\bar{Q}_{23}}{\partial x_{123}}\bar{Q}_{13}-\frac{\partial\bar{Q}_{13}}{\partial x_{123}}\bar{Q}_{23}=0.
\end{aligned}
\end{equation}
We call polynomials $E_i$, $E_{ij}$ \emph{corner polynomials}.
\end{theo}
\begin{proof}
The derivation of the corner equation \eqref{eq: corner i} from the three-leg equations \eqref{eq: aux1}, \eqref{eq: aux2} can be described as elimination of the variable $X$ between the latter two equations. Obviously, one obtains an equivalent equation by eliminating the variable $x$ between the multi-affine forms $Q_{ij}=0$ and $Q_{ik}=0$ of these same equations.
\end{proof}

We remark that each polynomial $E_\ell$ ($\ell\in\{1,2,3,12,23,13\}$) is of degree 2 with respect to the variable $x_\ell$, does not depend on the ``opposite'' variable, and is of degree one with respect to all other variables.

From now on, we will refer to \eqref{eq:cornersystem} as to the system of corner equations corresponding to quad-equations \eqref{eq:quadsystem}, or simply as to the system of corner equations. Our main interest is in the following crucial property of such systems.

\begin{defi}\label{def:consistency}
A system of corner equations is called \emph{consistent}, if it has the minimal possible rank 2, i.e., if exactly two of these equations are independent.
\end{defi}
In \cite{variational}, we already gave two different proofs of consistency of the systems of corner equations coming from integrable quad-equations. One of these proofs directly utilized 3D consistency of quad-equations. In the present paper, we provide additional insights in the algebraic structure of the systems of corner equations, which lead to new insights into the nature of their consistency, as well.

We start by establishing certain remarkable relations between corner equations.
\begin{prop}\label{th:fractionalideal}
For a system of corner polynomials from \eqref{eq:cornersystem} the following relations hold:
\begin{align}
\label{eq:fractional1}
\frac{\partial E_{3}}{\partial x_{23}}E_{2}-\frac{\partial E_{2}}{\partial x_{23}}E_{3}&=Q_{23}^{2,3} E_{1},\\ \label{eq:fractional2}
\frac{\partial E_{3}}{\partial x_{1}}E_{2}-\frac{\partial E_{2}}{\partial x_{1}}E_{3}&=Q_{23}^{2,3} E_{23},\\
\intertext{and} \label{eq:fractional3}
\frac{\partial E_{13}}{\partial x_{1}}E_{12}-\frac{\partial E_{12}}{\partial x_{1}}E_{13}&=\bar{Q}_{23}^{12,13} E_{23},\\ \label{eq:fractional4}
\frac{\partial E_{13}}{\partial x_{23}}E_{12}-\frac{\partial E_{12}}{\partial x_{23}}E_{13}&=\bar{Q}_{23}^{12,13} E_{1},
\end{align}
where $Q_{23}^{2,3}$ and $\bar{Q}_{23}^{12,13}$ are the biquadratic polynomials
\begin{align*}
Q_{23}^{2,3}&=Q_{23}\frac{\partial^{2}Q_{23}}{\partial x\partial x_{23}}-\frac{\partial Q_{23}}{\partial x}\frac{\partial Q_{23}}{\partial x_{23}},\\\intertext{and}
\bar{Q}_{23}^{12,13}&=\bar{Q}_{23}\frac{\partial^{2}\bar{Q}_{23}}{\partial x_{1}\partial x_{123}}-\frac{\partial \bar{Q}_{23}}{\partial x_{1}}\frac{\partial \bar{Q}_{23}}{\partial x_{123}}
\end{align*}
of the quad-equations $Q_{23}=0$ and $\bar{Q}_{23}=0$, respectively (superscripts indicate variables on which these biquadratic polynomials depend). Due to the symmetry of the system of quad-equations~\eqref{eq:quadsystem}, equations obtained from \eqref{eq:fractional1}--\eqref{eq:fractional4} by cyclic permutations of indices $(123)$ hold true, as well.
\end{prop}
\begin{proof}
The proof of \eqref{eq:fractional1} is obtained by a direct computation with expressions for $E_2$, $E_3$ given in \eqref{eq:cornersystem}.
The proof of \eqref{eq:fractional2} is substantially more involved. One starts with the observation that the four-leg equation $(E_3)$, as given in \eqref{eq: corner i} for $i=3$, can be alternatively obtained as a sum of the three-leg forms (centered at $x_3$) of the quad-equation $\bar{Q}_{12}(x_3,x_{13},x_{23},x_{123};\alpha_1,\alpha_2)=0$ and the tetrahedron equation $\bar{T}(x_1,x_2,x_3,x_{123})=0$, respectively:
\begin{align*}
& \psi(X_3,X_{13};\alpha_1)-\psi(X_3,X_{23};\alpha_2)-\phi(X_3,X_{123};\alpha_1,\alpha_2)=0,\\
& \phi(X_3,X_1;\alpha_3,\alpha_1)+\phi(X_3,X_2;\alpha_2,\alpha_3)+\phi(X_3,X_{123};\alpha_1,\alpha_2)=0.
\end{align*}
Therefore, the polynomial equation $E_3=0$ can be alternatively obtained by eliminating the variable $x_{123}$ between the polynomial equations $\bar{Q}_{12}=0$ and $\bar{T}=0$. Choosing a suitable normalization of the polynomial $\bar{T}$ (which is defined up to a constant factor), we can assume that
\begin{equation}\label{eq: proof1 aux1}
E_{3}=\frac{\partial \bar{T}}{\partial x_{123}}\bar{Q}_{12}-\frac{\partial\bar{Q}_{12}}{\partial x_{123}}\bar{T}.
\end{equation}
Analogously,
\begin{equation}\label{eq: proof1 aux2}
E_{2}=\frac{\partial\bar{Q}_{13}}{\partial x_{123}}\tilde{T}-
\frac{\partial \tilde{T}}{\partial x_{123}}\bar{Q}_{13},
\end{equation}
with $\widetilde{T}=\gamma\bar{T}$ being some other normalization of the same tetrahedron equation. Substituting expressions \eqref{eq: proof1 aux1}, \eqref{eq: proof1 aux2} into the left-hand side of \eqref{eq:fractional2}, one arrives after a straightforward computation, using that all polynomials $\bar{Q}_{12}$, $\bar{Q}_{13}$, $\bar{T}$ are multi-affine:
\begin{align*}
&\frac{\partial E_{3}}{\partial x_{1}}E_{2}-\frac{\partial E_{2}}{\partial x_{1}}E_{3}=
\gamma\bar{T}^{2,3}E_{23},\\
\intertext{with the biquadratic polynomial}
&\bar{T}^{2,3}=\bar{T}\frac{\partial^{2}\bar{T}}{\partial x_{1}\partial x_{123}}-
\frac{\partial\bar{T}}{\partial x_{1}}\frac{\partial{\bar{T}}}{\partial x_{123}}.
\end{align*}
It remains to prove that $\gamma\bar{T}^{2,3}=Q_{23}^{2,3}$. A straightforward computation with the expression \eqref{eq: proof1 aux1} for $E_3$ gives:
\[
\frac{\partial\bar{Q}_{12}}{\partial x_{23}}E_{3}-\frac{\partial E_{3}}{\partial x_{23}}\bar{Q}_{12}
=\bar{Q}_{12}^{3,13}\bar{T},
\]
which can be considered as a formula for the multi-affine polynomial $\bar{T}$. Computing its biquadratic $\bar{T}^{2,3}$ by a standard Wronskian-type operation eliminating the variables $x_{1}$, $x_{123}$ (on which the polynomial $\bar{Q}_{12}^{3,13}$ does not depend), we arrive at
\[
\bar{T}^{2,3}=\bar{\beta}^{-1} Q_{23}^{2,3},\qquad\text{where}\qquad \bar{\beta}:=\frac{\bar{Q}_{12}^{3,13}}{Q_{13}^{3,13}}.
\]
Similarly, a computation with the expression \eqref{eq: proof1 aux2} for $E_2$ gives:
\[
\frac{\partial E_{2}}{\partial x_{23}}\bar{Q}_{13}-\frac{\partial \bar{Q}_{13}}{\partial x_{23}}E_{2}=
\bar{Q}_{13}^{2,12}\tilde{T},
\]
and then we find an expression for the biquadratic $\tilde{T}^{2,3}$ of the multi-affine polynomial $\tilde{T}$:
\[
\tilde{T}^{2,3}=\tilde{\beta} Q_{23}^{2,3}\qquad\text{where}\qquad
 \tilde{\beta}:=\frac{Q_{12}^{2,12}}{\bar{Q}_{13}^{2,12}}.
\]
But, according to \cite[Proposition 5]{ABS1}, we have:
\[
\frac{Q_{12}^{2,12}}{\bar{Q}_{13}^{2,12}}=\frac{\bar{Q}_{12}^{3,13}}{Q_{13}^{3,13}},\quad\text{i.e.,}\quad \bar{\beta}=\tilde{\beta}=:\beta.
\]
Therefore, $\tilde{T}^{2,3}=\beta^{2}\bar{T}^{2,3}$, so that $\tilde{T}=\beta\bar{T}$ or $\tilde{T}=-\beta\bar{T}$. A case-by-case inspection of the ABS list shows that $\tilde{T}=\beta\bar{T}$, so $\beta=\gamma$, and therefore $Q_{23}^{2,3}=\gamma\bar{T}^{2,3}$, which finishes the proof of equation \eqref{eq:fractional2}.
The proofs of equations~\eqref{eq:fractional3} and \eqref{eq:fractional4} are analogous to the proofs of equations~\eqref{eq:fractional1} and \eqref{eq:fractional2}, respectively.
\end{proof}

We now recall the definition of fractional ideals, cf. \cite[section VII.1]{bourbaki}, adapted to our situation:
\begin{defi}[Fractional ideal]\label{def:fractionalideal}
Set $\mathcal{R}:=\C[x_{1},x_{2},x_{3},x_{12},x_{23},x_{13}]$ and $r\in\mathcal{R}\setminus\{0\}$. Then a \emph{fractional ideal $\mathcal{I}$ with the denominator $r$} is an $\mathcal{R}$-submodule of the field of fractions of $\mathcal{R}$ such that $r\mathcal{I}\subset\mathcal{R}$.
For three polynomials $p_{1},p_{2},r\in\mathcal{R}\setminus\{0\}$, the \emph{fractional ideal generated by $p_{1}$ and $p_{2}$ with the denominator $r$} is denoted by $\langle p_{1},p_{2}\rangle_{r}$ and consists of all polynomials $p$ representable as $rp=r_{1}p_{1}+r_{2}p_{2}$ with some $r_{1},r_{2}\in\mathcal{R}$. We say that the polynomial $p\in\mathcal{R}$ is in a fractional ideal generated by two polynomials $p_{1}, p_{2}\in\mathcal{R}$ if $p\in\langle p_{1},p_{2}\rangle_{r}$ for some $r\in\mathcal{R}$.
\end{defi}

With the help of this definition, we can re-phrase Proposition~\ref{th:fractionalideal} by saying that $E_2$, $E_3$, $E_{23}$ and $E_{13}$ are in a fractional ideal generated by $E_1$ and $E_{23}$. This easily yields the following more general statement which can be considered as a description of consistency of corner equations in terms of fractional ideals:
\begin{theo}\label{th:fractionalideal1}
Any two of the corner polynomials from the system~\eqref{eq:cornersystem} generate a fractional ideal which contains the remaining four corner polynomials. Therefore, the rank of the system~\eqref{eq:cornersystem} is 2.
\end{theo}

\section{Consistency of corner equations and octahedron relations}\label{sec:octs}

In view of Theorem~\ref{th:fractionalideal1}, it is natural to inquire whether one can find two relations which would be either simpler or more symmetric than the corner equations themselves but which would still contain the whole information of the system~\eqref{eq:cornersystem}, for instance in the sense that all corner equations are in a fractional ideal generated by those two relations. We will see that one can find such two relations in the multi-affine class.
\begin{defi}[Octahedron relation]\label{def:ocathedron}
A relation $R=0$, where $R\in\C[x_{1},x_{2},x_{3},x_{12},x_{23},x_{13}]$ is a multi-affine polynomial, is an \emph{octahedron relation} of a system of corner equations~\eqref{eq:cornersystem} if $R$ belongs to a fractional ideal generated by (any two of) the corner polynomials $E_i$, $E_{ij}$.
\end{defi}
Immediately from this definition there follows:
\begin{theo}\label{th:octs}
Suppose the system of corner equations~\eqref{eq:cornersystem} admits two octahedron relations $R_{1}=0$ and $R_{2}=0$, where $R_1$, $R_2$ are two linearly independent irreducible polynomials. Then all corner polynomials are in a fractional ideal generated by $R_{1}$ and $R_{2}$, which is equivalent to the consistency of the system~\eqref{eq:cornersystem}.
\end{theo}
\begin{rem}
Actually, one can give a more precise formulation: the corner polynomials $E_i$, $E_{jk}$ divide the polynomials
\[
B_i:=\frac{\partial R_2}{\partial x_{jk}}R_1-\frac{\partial R_1}{\partial x_{jk}}R_2, \qquad \text{resp.}\qquad
B_{jk}:=\frac{\partial R_2}{\partial x_i}R_1-\frac{\partial R_1}{\partial x_i}R_2.
\]
To show this for $E_{1}=0$, say (since for all other corner polynomials the proof is analogous), we observe that the polynomial $B_1$ is a non-zero polynomial, independent of $x_{23}$ and of degree 2 with respect to all other variables $x_{1}$, $x_{2}$, $x_{3}$, $x_{12}$ and $x_{13}$. By definition, $B_1$ is in a fractional ideal generated by $E_1$ and $E_{23}$. The coefficient by $E_{23}$ in the corresponding representation $rB=p_1E_1+p_{23}E_{23}$ must vanish, since both $B_1$ and $E_1$ do not depend on $x_{23}$. Therefore, $B_1=0$ as soon as $E_1=0$, hence $B_1$ is divisible by $E_1$.
\end{rem}

One of the main results of this paper is the following.
\begin{theo}[Consistency of corner equations]\label{th:consistency1}
A system of corner equations~\eqref{eq:cornersystem} coming from any of the ABS quad-equations admits two independent octahedron relations.
\end{theo}

We will prove this by an explicit construction of the octahedron relations.

\section{First octahedron relation}\label{sect: 1st oct}

The first octahedron relation can be found in the following way:
\begin{theo}\label{th:oct1}
Equation $R_{1}=0$ with
\begin{equation*}
R_{1}:=\frac{1}{2}\left(\frac{\partial E_{1}}{\partial x_{1}}+\frac{\partial E_{23}}{\partial x_{23}}\right)
\end{equation*}
is an octahedron relation of the system~\eqref{eq:cornersystem}. Analogous octahedron relations are obtained through cyclic permutations of indices $(123)$:
\begin{equation*}
R_{2}:=\frac{1}{2}\left(\frac{\partial E_{2}}{\partial x_{2}}+\frac{\partial E_{13}}{\partial x_{13}}\right),\qquad
R_{3}:=\frac{1}{2}\left(\frac{\partial E_{3}}{\partial x_{3}}+\frac{\partial E_{12}}{\partial x_{12}}\right).
\end{equation*}
\end{theo}
\begin{proof}
Since polynomials $E_1$ and $E_{23}$ are of degree 2 with respect to the variables $x_1$, resp. $x_{23}$, and of degree 1 with respect to all other variables, it follows that polynomial $R_1$ is multi-affine. Using equations~\eqref{eq:fractional1} and \eqref{eq:fractional2}, we find:
\begin{align*}
\frac{\partial^{2} E_{3}}{\partial x_{1}\partial x_{23}}E_{2}-\frac{\partial^{2} E_{2}}{\partial x_{1}\partial x_{23}}E_{3} & =\frac{1}{2}\left(\frac{\partial}{\partial x_{1}}\left(\frac{\partial E_{3}}{\partial x_{23}}E_{2}-\frac{\partial E_{2}}{\partial x_{23}}E_{3}\right)+\frac{\partial}{\partial x_{23}}\left(\frac{\partial E_{3}}{\partial x_{1}}E_{2}-\frac{\partial E_{2}}{\partial x_{1}}E_{3}\right)\right)\\
& =\frac{1}{2}Q_{23}^{2,3}\left(\frac{\partial E_{1}}{\partial x_{1}}+\frac{\partial E_{23}}{\partial x_{23}}\right)
=Q_{23}^{2,3}R_{1}.
\end{align*}
Therefore, $R_1$ is in a fractional ideal generated by $E_{2}$ and $E_{3}$ (and therefore in a fractional ideal generated by any other two corner equations from \eqref{eq:cornersystem}).
\end{proof}
The following proposition is based on straightforward case-by-case computations for all ABS equations:
\begin{prop}\label{obs:symmetry}
\quad

\begin{itemize}
\item In the case coming from the quad-equation $A_{2}$, the following equation is satisfied:
\[
\frac{\sin(\alpha_{2}-\alpha_{3})}{\sin(\alpha_{1})}R_{1}+
\frac{\sin(\alpha_{3}-\alpha_{1})}{\sin(\alpha_{2})}R_{2}+
\frac{\sin(\alpha_{1}-\alpha_{2})}{\sin(\alpha_{3})}R_{3}=0.
\]
\item In the case coming from the quad-equation $Q_{4}$, the following equation is satisfied:
\[
\sn(\alpha_{2}-\alpha_{3})R_{1}+\sn(\alpha_{3}-\alpha_{1})R_{2}+\sn(\alpha_{1}-\alpha_{2})R_{3}=0.
\]
\item In these two cases, any two of the polynomials $R_1$, $R_2$, $R_3$ are linearly independent.
\item In all other cases we have:
\[
\frac{R_{1}}{f_{1}}=\frac{R_{2}}{f_{2}}=\frac{R_{3}}{f_{3}}=:\Omega_1,
\]
with the constant factors
\[
f_1:=\frac{\partial^2 Q_{12}}{\partial x\partial x_1}=-\frac{\partial^2 Q_{13}}{\partial x\partial x_1},\qquad
f_2:=\frac{\partial^2 Q_{23}}{\partial x\partial x_1}=-\frac{\partial^2 Q_{12}}{\partial x\partial x_1},\qquad
f_3:=\frac{\partial^2 Q_{13}}{\partial x\partial x_3}=-\frac{\partial^2 Q_{23}}{\partial x\partial x_3}.
\]
\end{itemize}
\end{prop}

Thus, in the cases coming from $A_{2}$ and $Q_{4}$, all corner equations from the system~\eqref{eq:cornersystem} are in a fractional ideal generated by $R_{1}$ and $R_{2}$. In all other cases, one has just one octahedron relation $\Omega_1=0$, symmetric with respect to cyclic permutations of indices $(123)$, and one has to look for further octahedron relations. A list of the polynomials $R_{1}$, resp. $\Omega_1$, can be found in Appendix~\ref{sec:octlist}.

Before we turn to the problem of finding the second independent octahedron relation for equations not coming from case $A_2$ and $Q_4$, we mention another expression for the octahedron relation $\Omega_{1}=0$.
\begin{prop}\label{th:summation}
In the cases coming from the quad-equations $Q_{1}^{\delta}$, $Q_{2}$, $Q_{3}^{\delta}$ and $A_{1}^{\delta}$ the octahedron relation $\Omega_{1}=0$ can be written as
\[
\Omega_{1}=Q_{12}+Q_{23}+Q_{13}+T=0,
\]
where $T=T(x,x_{12},x_{23},x_{13})=0$ is the tetrahedron equation. In the cases coming from the quad-equations $H_{1}$, $H_{2}$ and $H_{3}^{\delta}$ the octahedron relation $\Omega_{1}=0$ can be written as
\[
\Omega_{1}=Q_{12}+Q_{23}+Q_{13}=0.
\]
\end{prop}
\begin{proof}
A straightforward computation gives:
\[
\Omega_{1}=\frac{1}{2f_{1}}\left(\frac{\partial E_{1}}{\partial x_{1}}+\frac{\partial E_{23}}{\partial x_{23}}\right)
=Q_{13}+Q_{12}+G_1,
\]
where
\begin{align*}
G_{1} & = \frac{1}{2f_{1}}\left(\frac{\partial^{2}Q_{13}}{\partial x\partial x_{1}}Q_{12}-\frac{\partial^{2}Q_{12}}{\partial x\partial x_{1}}Q_{13}+\frac{\partial Q_{12}}{\partial x}\frac{\partial Q_{13}}{\partial x_{1}}-\frac{\partial Q_{13}}{\partial x}\frac{\partial Q_{12}}{\partial x_{1}}\right)\\
& \phantom{=\ } +\frac{1}{2f_{1}}\left(\frac{\partial^{2}\bar{Q}_{13}}{\partial x_{23}\partial x_{123}}\bar{Q}_{12}-\frac{\partial^{2}\bar{Q}_{12}}{\partial x_{23}\partial x_{123}}\bar{Q}_{13}+\frac{\partial\bar{Q}_{13}}{\partial x_{123}}\frac{\partial\bar{Q}_{12}}{\partial x_{23}}-\frac{\partial\bar{Q}_{12}}{\partial x_{123}}\frac{\partial\bar{Q}_{13}}{\partial x_{23}}\right).
\end{align*}
Using the fact that all polynomials $Q_{ij}$, $\bar{Q}_{ij}$ are multi-affine, we immediately compute that
\[
\frac{\partial G_{1}}{\partial x_{1}}=0, \qquad \frac{\partial G_{1}}{\partial x_{123}}=0.
\]
Therefore, $G_{1}=G_{1}(x,x_{2},x_{3},x_{12},x_{23},x_{13})$ is multi-affine and independent of $x_{1}$ and $x_{123}$. Analogously,
\[
\Omega_{1}=\frac{1}{2f_{2}}\left(\frac{\partial E_{2}}{\partial x_{2}}+\frac{\partial E_{13}}{\partial x_{13}}\right)=Q_{12}+Q_{23}+G_2(x,x_{3},x_{1},x_{12},x_{23},x_{13}),
\]
where $G_2(x,x_{3},x_{1},x_{12},x_{23},x_{13})$ is multi-affine and independent of $x_{2}$ and $x_{123}$, and
\[
\Omega_{1}=\frac{1}{2f_{3}}\left(\frac{\partial E_{3}}{\partial x_{3}}+\frac{\partial E_{12}}{\partial x_{12}}\right)=Q_{23}+Q_{13}+G_3(x,x_{1},x_{2},x_{12},x_{23},x_{13}),
\]
where $G_3(x,x_{1},x_{2},x_{12},x_{23},x_{13})$ is multi-affine and independent of $x_{3}$ and $x_{123}$. Thus,
\[
Q_{13}+Q_{12}+G_{1}=Q_{12}+Q_{23}+G_{2}=Q_{23}+Q_{13}+G_{3}.
\]
Set
\[
T:=G_{1}-Q_{23}=G_{2}-Q_{13}=G_{3}-Q_{12}.
\]
The first expression for $T$ shows that it is independent of $x_{1}$, the second shows that it is independent of $x_{2}$, and the third shows that it is independent of $x_{3}$. Therefore,
\[
\Omega_{1}=Q_{12}+Q_{23}+Q_{13}+T(x,x_{12},x_{23},x_{13}),
\]
where $T(x,x_{12},x_{23},x_{13})$ is multi-affine and independent of $x_{1}$, $x_{2}$, $x_{3}$ and $x_{123}$. One shows by a simple case-by-case that in the cases coming from $Q_{1}^{\delta}$, $Q_{2}$, $Q_{3}^{\delta}$ and $A_{1}^{\delta}$, the polynomial $T$ defines the tetrahedron relation $T=0$, while in the cases coming from $H_{1}$, $H_{2}$ and $H_{3}^{\delta}$ we have $T\equiv0$.
\end{proof}

\begin{rem}
Of course, in the cases coming from $Q_{1}^{\delta}$, $Q_{2}$, $Q_{3}^{\delta}$ and $A_{1}^{\delta}$ one can write the octahedron relation $\Omega_{1}=0$ as
\[
\Omega_{1}=\bar{Q}_{12}+\bar{Q}_{23}+\bar{Q}_{13}+\bar{T}=0,
\]
where the polynomial $\bar{T}=\bar{T}(x_1,x_2,x_3,x_{123})$ defines the tetrahedron relation $\bar{T}=0$, while
in the cases coming from $H_{1}$, $H_{2}$ and $H_{3}^{\delta}$ we have
\[
\Omega_{1}=\bar{Q}_{12}+\bar{Q}_{23}+\bar{Q}_{13}=0.
\]
\end{rem}

\section{Second octahedron relation}\label{sect: 2nd oct}

Turning to the problem of finding the second octahedron relation in the cases different from $A_2$ and $Q_4$, we have the following results.
\begin{theo}\label{th:oct2}
The polynomials
\[
P_{1}:=E_{1}-x_{1}R_{1}\quad\text{and}\quad P_{23}:=E_{23}-x_{23}R_{1}
\]
satisfy the following identities:
\begin{equation}\label{eq:fractional P1}
\frac{\partial R_{1}}{\partial x_{23}}P_{1}-\frac{\partial P_{1}}{\partial x_{23}}R_{1}=
\frac{\partial R_{1}}{\partial x_{23}}E_{1},\qquad
\frac{\partial R_{1}}{\partial x_{1}}P_{1}-\frac{\partial P_{1}}{\partial x_{1}}R_{1}=
\frac{\partial R_{1}}{\partial x_{23}}E_{23},
\end{equation}
and, similarly,
\begin{equation}\label{eq:fractional P23}
\frac{\partial R_{1}}{\partial x_{23}}P_{23}-\frac{\partial P_{23}}{\partial x_{23}}R_{1}=
\frac{\partial R_{1}}{\partial x_{1}}E_{1}, \qquad
\frac{\partial R_{1}}{\partial x_{1}}P_{23}-\frac{\partial P_{23}}{\partial x_{1}}R_{1}=
\frac{\partial R_{1}}{\partial x_{1}}E_{13}.
\end{equation}
Thus, $P_1=0$ and $P_{23}=0$ are octahedron relations of the system~\eqref{eq:cornersystem}. One find further octahedron relations under cyclic permutations of indices $(123)$.
\end{theo}
\begin{proof}
We prove relations \eqref{eq:fractional P1} for $P_1$, the proof of relations \eqref{eq:fractional P23} for $P_{23}$ being similar. The first of the relations in \eqref{eq:fractional P1} is trivial:
\[
\frac{\partial R_{1}}{\partial x_{23}}P_{1}-\frac{\partial P_{1}}{\partial x_{23}}R_{1}=\frac{\partial R_{1}}{\partial x_{23}}(E_{1}-x_{1}R_{1})+x_{1}\frac{\partial R_{1}}{\partial x_{23}}R_{1}=\frac{\partial R_{1}}{\partial x_{23}}E_{1}.
\]
As for the second relation in \eqref{eq:fractional P1}, we compute:
\begin{multline*}
\frac{\partial R_{1}}{\partial x_{1}}P_{1}-\frac{\partial P_{1}}{\partial x_{1}}R_{1}-\frac{\partial R_{1}}{\partial x_{23}}E_{23}=\frac{1}{2}\frac{\partial^{2}E_{1}}{\partial x_{1}^{2}}\left(E_{1}-\frac{1}{2}\left(\frac{\partial E_{1}}{\partial x_{1}}+\frac{\partial E_{23}}{\partial x_{23}}\right)\right)\\
-\frac{1}{2}\left(\frac{\partial E_{1}}{\partial x_{1}}-\frac{1}{2}\left(\frac{\partial E_{1}}{\partial x_{1}}+\frac{\partial E_{23}}{\partial x_{23}}\right)-\frac{1}{2}x_{1}\frac{\partial^{2}E_{1}}{\partial x_{1}^{2}}\right)\left(\frac{\partial E_{1}}{\partial x_{1}}+\frac{\partial E_{23}}{\partial x_{23}}\right)-\frac{1}{2}\frac{\partial^{2}E_{23}}{\partial x_{23}^{2}}E_{23}\\
=\frac{1}{2}\left(E_{1} \frac{\partial^{2} E_{1}}{\partial x_{1}^{2}}-\frac{1}{2}\left(\frac{\partial E_{1}}{\partial x_{1}}\right)^{2}-E_{23} \frac{\partial^{2} E_{23}}{\partial x_{23}^{2}}+\frac{1}{2}\left(\frac{\partial E_{23}}{\partial x_{23}}\right)^{2}\right).
\end{multline*}
This vanishes due to following Lemma~\ref{lem:discr}.
\end{proof}

\begin{lemma}\label{lem:discr}
The discriminants of the opposite corner polynomials with respect to their central points coincide:
\[
E_{1} \frac{\partial^{2} E_{1}}{\partial x_{1}^{2}}-\frac{1}{2}\left(\frac{\partial E_{1}}{\partial x_{1}}\right)^{2}=E_{23} \frac{\partial^{2} E_{23}}{\partial x_{23}^{2}}-\frac{1}{2}\left(\frac{\partial E_{23}}{\partial x_{23}}\right)^{2}.
\]
\end{lemma}
\begin{proof}
We have
{\allowdisplaybreaks
\begin{multline*}
E_{1} \frac{\partial^{2} E_{1}}{\partial x_{1}^{2}}-\frac{1}{2}\left(\frac{\partial E_{1}}{\partial x_{1}}\right)^{2}=2\left(\frac{\partial Q_{12}}{\partial x}Q_{13}-\frac{\partial Q_{13}}{\partial x}Q_{12}\right)\left(\frac{\partial^{2}Q_{12}}{\partial x\partial x_{1}}\frac{\partial Q_{13}}{\partial x_{1}}-\frac{\partial^{2}Q_{13}}{\partial x\partial x_{1}}\frac{\partial Q_{12}}{\partial x_{1}}\right)\\
-\frac{1}{2}\left(\frac{\partial^{2}Q_{12}}{\partial x\partial x_{1}}Q_{13}+\frac{\partial Q_{12}}{\partial x}\frac{\partial Q_{13}}{\partial x_{1}}-\frac{\partial^{2}Q_{13}}{\partial x\partial x_{1}}Q_{12}-\frac{\partial Q_{13}}{\partial x}\frac{\partial Q_{12}}{\partial x_{1}}\right)^{2}\\
=\frac{\partial Q_{13}}{\partial x_{1}}Q_{13}\frac{\partial^{2}Q_{12}}{\partial x\partial x_{1}}\frac{\partial Q_{12}}{\partial x}-2\frac{\partial^{2}Q_{13}}{\partial x\partial x_{1}}Q_{13}\frac{\partial Q_{12}}{\partial x}\frac{\partial Q_{12}}{\partial x_{1}}-2\frac{\partial Q_{13}}{\partial x}\frac{\partial Q_{13}}{\partial x_{1}}\frac{\partial^{2}Q_{12}}{\partial x\partial x_{1}}Q_{12}\\
+\frac{\partial^{2} Q_{13}}{\partial x\partial x_{1}}\frac{\partial Q_{13}}{\partial x}\frac{\partial Q_{12}}{\partial x_{1}}Q_{12}-\frac{1}{2}\left(Q_{13}\frac{\partial^{2}Q_{12}}{\partial x\partial x_{1}}\right)^{2}-\frac{1}{2}\left(\frac{\partial Q_{13}}{\partial x_{1}}\frac{\partial Q_{12}}{\partial x}\right)^{2}-\frac{1}{2}\left(\frac{\partial^{2}Q_{13}}{\partial x\partial x_{1}}Q_{12}\right)^{2}\\
-\frac{1}{2}\left(\frac{\partial Q_{13}}{\partial x}\frac{\partial Q_{12}}{\partial x_{1}}\right)^{2}+\frac{\partial^{2}Q_{13}}{\partial x\partial x_{1}}Q_{13}\frac{\partial^{2}Q_{12}}{\partial x\partial x_{1}}Q_{12}+\frac{\partial Q_{13}}{\partial x}Q_{13}\frac{\partial^{2}Q_{12}}{\partial x\partial x_{1}}\frac{\partial Q_{12}}{\partial x_{1}}\\
+\frac{\partial^{2}Q_{13}}{\partial x\partial x_{1}}\frac{\partial Q_{13}}{\partial x_{1}}\frac{\partial Q_{12}}{\partial x}Q_{12}+\frac{\partial Q_{13}}{\partial x}\frac{\partial Q_{13}}{\partial x_{1}}\frac{\partial Q_{12}}{\partial x}\frac{\partial Q_{12}}{\partial x_{1}}.
\end{multline*}
}
This expression is manifestly symmetric with respect to the permutation $x\leftrightarrow x_{1}$. Therefore,
\[
E_{1} \frac{\partial^{2} E_{1}}{\partial x_{1}^{2}}-\frac{1}{2}\left(\frac{\partial E_{1}}{\partial x_{1}}\right)^{2}=F_{1} \frac{\partial^{2} F_{1}}{\partial x^{2}}-\frac{1}{2}\left(\frac{\partial F_{1}}{\partial x}\right)^{2},
\]
where
\[
F_{1}:=\frac{\partial Q_{12}}{\partial x_{1}}Q_{13}-\frac{\partial Q_{13}}{\partial x_{1}}Q_{12}.
\]
It remains to prove that
\begin{equation}\label{eq: proof2 aux1}
F_{1} \frac{\partial^{2} F_{1}}{\partial x^{2}}-\frac{1}{2}\left(\frac{\partial F_{1}}{\partial x}\right)^{2}=E_{23} \frac{\partial^{2} E_{23}}{\partial x_{23}^{2}}-\frac{1}{2}\left(\frac{\partial E_{23}}{\partial x_{23}}\right)^{2}.
\end{equation}
For this aim, we proceed similarly to the proof of Proposition \ref{th:fractionalideal}. One can represent the equation $F_1=0$ as the result of eliminating the variable $x_{23}$ between $Q_{23}=0$ and the suitably normalized tetrahedron equation $T=T(x,x_{12},x_{23},x_{13})=0$:
\[
F_{1}=\frac{\partial Q_{23}}{\partial x_{23}}T-\frac{\partial T}{\partial x_{23}}Q_{23}.
\]
On the other hand, we have a similar representation of the corner equation $E_{23}=0$ as the result of eliminating the variable $x$ between $Q_{23}=0$ and the possibly differently normalized tetrahedron equation $T=0$:
\[
E_{23}=\beta\left(\frac{\partial Q_{23}}{\partial x}T-\frac{\partial T}{\partial x}Q_{23}\right).
\]
A direct case-by-case check shows that for all ABS equations we have $\beta^{2}=1$. This yields \eqref{eq: proof2 aux1} by a straightforward computation.
\end{proof}

The following proposition is proved by case-by-case computations.

\begin{prop}\label{obs:oct2}
For systems of corner equations~\eqref{eq:cornersystem} coming from $Q_{1}^{\delta}$, $Q_{2}$, $Q_{3}^{\delta}$, $H_{1}$, $H_{2}$, $H_{3}^{\delta}$ or $A_{1}^{\delta}$, the octahedron relation $\Omega_2=0$ with
\[
\Omega_{2}:=\frac{1}{g_{1}}(P_{1}-P_{23})
\]
is symmetric with respect to the cyclic permutation $(123)$ of indices. Here, $g_1$ is a constant factor given by
\[
g_{1}:=\begin{cases}
f_{1}(\alpha_2+\alpha_3-\alpha_1)& \text{if \eqref{eq:cornersystem} comes from $Q_{1}^{\delta}$ or $Q_{2}$,}\\
2f_1 e^{i(\alpha_1+\alpha_2+\alpha_3)}\sin\left(\frac{1}{2}(\alpha_{2}+\alpha_{3}-\alpha_1)\right) &
\text{if \eqref{eq:cornersystem} comes from $Q_{3}^{\delta}$,}\\
f_{1}& \text{if \eqref{eq:cornersystem} comes from $H_{1}$, $H_{2}$ or $A_{1}^{\delta}$,}\\
1& \text{if \eqref{eq:cornersystem} comes from $H_{3}^{\delta}$.}
\end{cases}
\]
\end{prop}

The list of the polynomials $\Omega_{2}$ can be found in Appendix~\ref{sec:octlist}.

From Theorem~\ref{th:oct2} and Proposition~\ref{obs:oct2} we get the following corollary:
\begin{cor}
Corner equations~\eqref{eq:cornersystem} coming from $Q_{1}^{\delta}$, $Q_{2}$, $Q_{3}^{\delta}$, $H_{1}$, $H_{2}$, $H_{3}^{\delta}$ or $A_{1}^{\delta}$, are expressed through the octahedron equations $\Omega_{1}=0$ and $\Omega_{2}=0$ as follows:
\begin{align*}
\frac{\partial \Omega_{1}}{\partial x_{23}}\Omega_{2}-\frac{\partial \Omega_{2}}{\partial x_{23}}\Omega_{1}&=
g_{1}\left(\frac{\partial \Omega_{1}}{\partial x_{23}}-\frac{\partial \Omega_{1}}{\partial x_{1}}\right)E_{1},\\
\frac{\partial \Omega_{1}}{\partial x_{1}}\Omega_{2}-\frac{\partial \Omega_{2}}{\partial x_{1}}\Omega_{1}&=
g_{1}\left(\frac{\partial \Omega_{1}}{\partial x_{23}}-\frac{\partial \Omega_{1}}{\partial x_{1}}\right)E_{23}.
\end{align*}
Analogous formulas hold under cyclic permutations of indices $(123)$.
\end{cor}
These formulas give a concrete realization of the remark after Theorem~\ref{th:octs} for the polynomials $\Omega_1$, $\Omega_2$.

\section{Comparison with previously known results}

In our recent paper~\cite{variational} we already gave a list of single octahedron relations for several systems of corner equations, namely, in the cases coming from $Q_{1}^{\delta}$, $Q_{3}^{0}$, $H_{1}$, $H_{2}$ and $H_{3}^{\delta}$.

In the cases coming from $H_{1}$, $H_{2}$ and $H_{3}^{\delta}$, they are equivalent to the relations $\Omega_{1}=0$ found in the present paper.

In the case coming from $Q_{1}^{0}$, we presented in ~\cite{variational} the octahedron relation
\begin{equation}\label{eq:multiver}
\frac{(x_{12}-x_{1})(x_{23}-x_{2})(x_{13}-x_{3})}{(x_{12}-x_{2})(x_{23}-x_{3})(x_{13}-x_{1})}=1
\end{equation}
(equation ($\chi_2$) in the classification of~\cite{octclassi}), which turns out to be equivalent to $\Omega_{2}=0$ found in the present paper.

In the case coming from $Q_{1}^{1}$ we presented in~\cite{variational} the octahedron relation resembling equation~\eqref{eq:multiver}:
\[
\frac{(x_{12}-x_{1}+\alpha_2)(x_{23}-x_{2}+\alpha_3)(x_{13}-x_{3}+\alpha_1)}
{(x_{12}-x_{2}+\alpha_1)(x_{23}-x_{3}+\alpha_2)(x_{13}-x_{1}+\alpha_3)}=1.
\]
This is equivalent to $P^+=0$, where
\begin{align*}
P^{+}&= (x_{12}-x_{2}+\alpha_{1})(x_{23}-x_{3}+\alpha_{2})(x_{13}-x_{1}+\alpha_{3})\\
 &\phantom{=\ } -(x_{12}-x_{1}+\alpha_{2})(x_{23}-x_{2}+\alpha_{3})(x_{13}-x_{3}+\alpha_{1}).
\end{align*}
Another octahedron relation $P^-=0$ could be obtained by inverting the signs of all parameters $\alpha_i$ (this operation leaves corner equations invariant):
\begin{align*}
P^{-}&=  (x_{12}-x_{2}-\alpha_{1})(x_{23}-x_{3}-\alpha_{2})(x_{13}-x_{1}-\alpha_{3})\\
 &\phantom{=\ } -(x_{12}-x_{1}-\alpha_{2})(x_{23}-x_{2}-\alpha_{3})(x_{13}-x_{3}-\alpha_{1}).
\end{align*}
It turns out that in this case we have the following relations:
\[
\Omega_{1}=\frac{1}{2}(P^{+}+P^{-}),\qquad \Omega_{2}=\frac{1}{2}(P^{+}-P^{-}).
\]

In the case coming from $Q_{3}^{0}$ we gave an octahedron relation equivalent to $P^+=0$, where
\begin{align*}
P^{+} &= \left(e^{\frac{i}{2}(\alpha_{1}-\alpha_{2}-\alpha_{3})}x_{1}+
e^{-\frac{i}{2}(\alpha_{1}-\alpha_{2}-\alpha_{3})}x_{23}\right)(x_{2}x_{13}-x_{3}x_{12})\\
 &\phantom{=\ } +\left(e^{\frac{i}{2}(\alpha_{2}-\alpha_{3}-\alpha_{1})}x_{2}+
e^{-\frac{i}{2}(\alpha_{2}-\alpha_{3}-\alpha_{1})}x_{13}\right)(x_{3}x_{12}-x_{1}x_{23})\\
 &\phantom{=\ } +\left(e^{\frac{i}{2}(\alpha_{3}-\alpha_{1}-\alpha_{2})}x_{3}+
e^{-\frac{i}{2}(\alpha_{3}-\alpha_{1}-\alpha_{2})}x_{12}\right)(x_{1}x_{23}-x_{2}x_{13}).
\end{align*}
Again, another octahedron relation $P^-=0$ could be obtained by inverting the signs of all parameters $\alpha_i$:
\begin{align*}
P^{-}&= -\left(e^{-\frac{i}{2}(\alpha_{1}-\alpha_{2}-\alpha_{3})}x_{1}+
e^{\frac{i}{2}(\alpha_{1}-\alpha_{2}-\alpha_{3})}x_{23}\right)(x_{2}x_{13}-x_{3}x_{12})\\
 &\phantom{=\ } -\left(e^{-\frac{i}{2}(\alpha_{2}-\alpha_{3}-\alpha_{1})}x_{2}+
e^{\frac{i}{2}(\alpha_{2}-\alpha_{3}-\alpha_{1})}x_{13}\right)(x_{3}x_{12}-x_{1}x_{23})\\
 &\phantom{=\ } -\left(e^{-\frac{i}{2}(\alpha_{3}-\alpha_{1}-\alpha_{2})}x_{3}+
e^{\frac{i}{2}(\alpha_{3}-\alpha_{1}-\alpha_{2})}x_{12}\right)(x_{1}x_{23}-x_{2}x_{13}).
\end{align*}
It turns out that these polynomials $P^+$, $P^-$ are related to $\Omega_1$, $\Omega_2$ found in the present paper through the following formulas:
\[
\Omega_{2}=\frac{1}{2}(P^{+}-P^{-}), \qquad
(2 x_{1}x_{23}-x_{2}x_{13}-x_{3}x_{13})i \Omega_{1}=p^{+}P^{+}+p^{-}P^{-},
\]
where
\begin{align*}
p^{+} & =-e^{-\frac{i}{2}(\alpha_{1}-\alpha_{2}-\alpha_{3})}x_{1}+
\tfrac{1}{2}e^{-\frac{i}{2}(\alpha_{2}-\alpha_{3}-\alpha_{1})}x_{2}+
\tfrac{1}{2}e^{-\frac{i}{2}(\alpha_{3}-\alpha_{1}-\alpha_{2})}x_{3}\\
 &\phantom{=\ } +e^{\frac{i}{2}(\alpha_{1}-\alpha_{2}-\alpha_{3})}x_{23}-
 \tfrac{1}{2}e^{\frac{i}{2}(\alpha_{2}-\alpha_{3}-\alpha_{1})}x_{13}-
\tfrac{1}{2}e^{\frac{i}{2}(\alpha_{3}-\alpha_{1}-\alpha_{2})}x_{12}\\
\intertext{and}
p^{+}&=e^{\frac{i}{2}(\alpha_{1}-\alpha_{2}-\alpha_{3})}x_{1}-
\tfrac{1}{2}e^{\frac{i}{2}(\alpha_{2}-\alpha_{3}-\alpha_{1})}x_{2}-
\tfrac{1}{2}e^{\frac{i}{2}(\alpha_{3}-\alpha_{1}-\alpha_{2})}x_{3}\\
 &\phantom{=\ } -e^{-\frac{i}{2}(\alpha_{1}-\alpha_{2}-\alpha_{3})}x_{23}+
\tfrac{1}{2}e^{-\frac{i}{2}(\alpha_{2}-\alpha_{3}-\alpha_{1})}x_{13}+
\tfrac{1}{2}e^{-\frac{i}{2}(\alpha_{3}-\alpha_{1}-\alpha_{2})}x_{12}.
\end{align*}
This shows explicitly that $P^{+}$ and $P^{-}$ are in a fractional generated by $\Omega_{1}$ and $\Omega_{2}$.

In~\cite{variational} we also discussed the system of corner equations
\begin{align}
& \frac{\alpha_{j}x_{i}-x_{ij}}{x_{i}}\cdot
\frac{\alpha_{i}x_{i}-\alpha_{j}x_{j}}{\alpha_{j}x_{i}-\alpha_{i}x_{j}}\cdot
\frac{x_{i}}{\alpha_{k}x_{i}-x_{ik}}\cdot
\frac{\alpha_{k}x_{i}-\alpha_{i}x_{k}}{\alpha_{i}x_{i}-\alpha_{k}x_{k}}=1,
\label{eq:asymmetric i} \tag{$E_i$}\\
& \frac{\alpha_{j}x_{i}-x_{ij}}{x_{i}}\cdot
\frac{\alpha_{j}x_{ij}-\alpha_{k}x_{ik}}{\alpha_{k}x_{ij}-\alpha_{j}x_{ik}}\cdot
\frac{x_{j}}{\alpha_{i} x_{j}-x_{ij}}\cdot
\frac{\alpha_{k}x_{ij}-\alpha_{i}x_{jk}}{\alpha_{i}x_{ij}-\alpha_{k}x_{jk}}=1
\label{eq:asymmetric ij} \tag{$E_{ij}$}
\end{align}
(where $(i,j,k)$ is a permutation of $(1,2,3)$), which cannot be derived from an integrable system of quad-equations. We gave the octahedron relation $P^+=0$, with
\[
P^{+}=x_{1}x_{2}(\alpha_{1}x_{13}-\alpha_{2}x_{23})+x_{2}x_{3}(\alpha_{2}x_{12}-\alpha_{3}x_{13})
+x_{3}x_{1}(\alpha_{3}x_{23}-\alpha_{1}x_{12})=0.
\]
This relation is of the type~($\chi_{4}$) in the classification from~\cite{octclassi}. Another octahedron relation $P^-=0$, which can also be obtained from the corresponding octahedron relation in the case coming from $Q_{3}^{0}$ using the limiting procedure described in~\cite{variational}, reads as follows:
\[
P^{-}=(\alpha_{2}x_{1}-x_{12})(\alpha_{3}x_{2}-x_{23})(\alpha_{1}x_{3}-x_{13})
-(\alpha_{1}x_{2}-x_{12})(\alpha_{2}x_{3}-x_{23})(\alpha_{3}x_{1}-x_{13}).
\]
It turns out that Theorem \ref{th:oct1} is perfectly applicable in this case, leading to a permutationally symmetric octahedron relation $\Omega_1=0$ with
\begin{align*}
\Omega_{1}&=  \ \frac{1}{2}\left(\frac{\partial E_{1}}{\partial x_{1}}+\frac{\partial E_{23}}{\partial x_{23}}\right)\\
& =  \ \alpha_{1}\alpha_{2}(x_{1}x_{13}-x_{2}x_{23})+
\alpha_{2}\alpha_{3}(x_{2}x_{12}-x_{3}x_{13})+
\alpha_{3}\alpha_{1}(x_{3}x_{23}-x_{1}x_{12})\\
&\phantom{=\ } +\alpha_{1}^{2}(\alpha_{3}x_{2}-\alpha_{2}x_{3})+
\alpha_{2}^{2}(\alpha_{1}x_{3}-\alpha_{3}x_{1})+
\alpha_{3}^{3}(\alpha_{2}x_{1}-\alpha_{1}x_{2}).
\end{align*}
Polynomial $\Omega_1$ lies in a fractional ideal generated by $P^+$, $P^-$, as the following formulas show:
\[
(2 x_{1}x_{23}-x_{2}x_{13}-x_{3}x_{12}) \Omega_{1}=p^{+}P^{+}+p^{-}P^{-},
\]
where
\begin{align*}
& p^{+}=2\alpha_{2}\alpha_{3}x_{1}-\alpha_{3}\alpha_{1}x_{2}-\alpha_{1}\alpha_{2}x_{3}+
2\alpha_{1}x_{23}-\alpha_{2}x_{13}-\alpha_{3}x_{12},\\
& p^{-}=2\alpha_{1}x_{1}-\alpha_{2}x_{2}-\alpha_{3}x_{3}.
\end{align*}
Constructions of Theorem~\ref{th:oct2} and Proposition~\ref{obs:oct2} are also applicable in this case and lead to further octahedron relations, which however are not permutationally symmetric (but, of course, lie in a fractional ideal generated by $P^{+}$ and $P^{-}$). For instance, polynomial from Proposition~\ref{obs:oct2} satisfies
\[
P_1-P_{23}=(\alpha_{1}\alpha_{3}-\alpha_{2})P^{+}+\alpha_{2}P^{-}.
\]

\section{Conclusion}\label{sec:concl}

We showed that every system of corner equations generated by a discrete 2-form corresponding to quad-equations from the ABS list can be encoded in a system of two octahedron relations: all corner equations follow from and therefore are satisfied by virtue of two octahedron equations. On the other hand, we found explicit and general formulae allowing us to express the octahedron relations in terms of the corner equations. Therefore, the system of corner equations and the system of two octahedron relations have to be seen on an equal footing from the algebraical point of view. This gives a new insight into the nature of consistency (integrability) of the system of corner equations and simultaneously poses a number of open questions.

Since corner equations are elementary building blocks of Euler-Lagrange equations of a pluri-Lagrangian problem (see~\cite{variational}), it is quite natural to inquire about the variational structure of the corresponding system of octahedron relations. This will be the subject of our ongoing research.

Furthermore, some of the octahedron relations are known to be integrable 3D equations themselves in the sense of multidimensional consistency (see~\cite{octclassi}). For instance, relation $\Omega_{2}=0$ in the case coming from $Q_{1}^{0}$, i.e., equation~\eqref{eq:multiver}, is the fundamental equation $(\chi_{2})$ in the classification in~\cite{octclassi}. Similarly, relation $\Omega_{1}=0$ in the case coming from $H_{1}$ is the equation $(\chi_{3})$ in this classification. It is an open problem whether they admit a variational structure. On the other hand, the majority of octahedron relations found in the present paper do not appear in the classification in~\cite{octclassi}, i.e., they are not integrable themselves in the sense of multi-dimensional consistency. It is not yet clear how they fit in the picture of integrability. Also this problem will be addressed in our future research.

There are two other octahedron type relations, i.e., relations of the type $$R(x_1,x_2,x_3,x_{12},x_{23},x_{13})=0,$$ satisfied by solutions of corner equations. One of them is the closure relation of the corresponding discrete two-form $\Ell$ (see~\cite{LN,BS1,variational}). The second one can be expressed in terms of the biquadratic polynomials associated with the quad-equations, and reads
\[
\frac{Q_{12}^{1,12}}{Q_{12}^{2,12}}\cdot\frac{Q_{23}^{2,23}}{Q_{23}^{3,23}}\cdot\frac{Q_{13}^{3,13}}{Q_{13}^{1,13}}=-1.
\]
(To prove the latter one for solutions of quad-equations, one uses the following identities:
\[
\frac{Q_{12}^{0,1}}{Q_{12}^{0,2}}=\frac{Q_{12}^{1,12}}{Q_{12}^{2,12}},\qquad \frac{Q_{23}^{0,2}}{Q_{23}^{0,3}}=\frac{Q_{23}^{2,23}}{Q_{23}^{3,23}},\qquad
\text{and}\qquad\frac{Q_{13}^{0,3}}{Q_{13}^{0,1}}=\frac{Q_{13}^{3,13}}{Q_{13}^{1,13}},
\]
which is \cite[formula (60)]{ABS1}, and
\[
\frac{Q_{12}^{0,1}}{Q_{12}^{0,2}}\cdot\frac{Q_{23}^{0,2}}{Q_{23}^{0,3}}\cdot\frac{Q_{13}^{0,3}}{Q_{13}^{0,1}}=-1,
\]
which is \cite[formula (16)]{ABS1}. Then one shows that the resulting relation holds true for solutions of corner equations, as well.)

However, both of them can not be written as $R=0$ with a multi-affine polynomial $R$, and therefore they do not qualify as octahedron relations in our sense.

\section*{Acknowledgment}
This research is supported by the DFG Collaborative Research Center TRR 109 ``Discretization in Geometry and Dynamics''.

\appendix
\section{ABS list of quad-equations}\label{sec:quadeqs}
In this section we give the list of polynomials $Q_{12}$ from the systems of quad-equations~\eqref{eq:quadsystem} for all cases we consider in this paper:
\begin{itemize}
\item[$Q_{1}^{\delta}$:] $Q_{12}=\alpha_{1}(xx_{1}+x_{2}x_{12})-\alpha_{2}(xx_{2}+x_{1}x_{12})-
    (\alpha_{1}-\alpha_{2})(xx_{12}+x_{1}x_{2})+\delta\alpha_{1}\alpha_{2}(\alpha_{1}-\alpha_{2})$
\item[$Q_{2}$:] $Q_{12}=\alpha_{1}(xx_{1}+x_{2}x_{12})-\alpha_{2}(xx_{2}+x_{1}x_{12})-
    (\alpha_{1}-\alpha_{2})(xx_{12}+x_{1}x_{2})$\\
$\phantom{Q_{12}=}+\alpha_{1}\alpha_{2}(\alpha_{1}-\alpha_{2})(x+x_{1}+x_{2}+x_{12})-
\alpha_{1}\alpha_{2}(\alpha_{1}-\alpha_{2})(\alpha_{1}^{2}-\alpha_{1}\alpha_{2}+\alpha_{2}^{2})$
\item[$Q_{3}^{\delta}$:] $Q_{12}=\sin(\alpha_{1})(xx_{1}+x_{2}x_{12})-\sin(\alpha_{2})(xx_{2}+x_{1}x_{12})-
    \sin(\alpha_{1}-\alpha_{2})(xx_{12}+x_{1}x_{2})$\\
$\phantom{Q_{12}=}+\delta\sin(\alpha_{1})\sin(\alpha_{2})\sin(\alpha_{1}-\alpha_{2})$
\item[$Q_{4}$:] $Q_{12}=\sn(\alpha_{1})(xx_{1}+x_{2}x_{12})-\sn(\alpha_{2})(xx_{2}+x_{1}x_{12})-
    \sn(\alpha_{1}-\alpha_{2})(xx_{12}+x_{1}x_{2})$\\
$\phantom{Q_{12}=}+\sn(\alpha_{1})\sn(\alpha_{2})\sn(\alpha_{1}-\alpha_{2})\left(1+k^{2}xx_{1}x_{2}x_{12}\right)$
\item[$H_{1}$:] $Q_{12}=(x-x_{12})(x_{1}-x_{2})-\alpha_{1}+\alpha_{2}$
\item[$H_{2}$:] $Q_{13}=(x-x_{12})(x_{1}-x_{2})-(\alpha_{1}-\alpha_{2})(x+x_{1}+x_{2}+x_{12})-\alpha_{1}^{2}+\alpha_{2}^{2}$
\item[$H_{3}^{\delta}$:] $Q_{12}=e^{\alpha_{1}}(xx_{1}+x_{2}x_{12})-e^{\alpha_{2}}(xx_{2}+x_{1}x_{12})+
    \delta\left(e^{2\alpha_{1}}-e^{2\alpha_{2}}\right)$
\item[$A_{1}^{\delta}$:] $Q_{12}=\alpha_{1}(xx_{1}+x_{2}x_{12})-\alpha_{2}(xx_{2}+x_{1}x_{12})+
    (\alpha_{1}-\alpha_{2})(xx_{12}+x_{1}x_{2})-\delta\alpha_{1}\alpha_{2}(\alpha_{1}-\alpha_{2})$
\item[$A_{2}$:] $Q_{12}=\sin(\alpha_{1})(xx_{2}+x_{1}x_{12})-\sin(\alpha_{2})(xx_{1}+x_{2}x_{12})-
    \sin(\alpha_{1}-\alpha_{2})(1+xx_{1}x_{2}x_{12})$
\end{itemize}
Here, $\delta\in\{0,1\}$ and $k$ is the modulus of $\sn(y)=\sn(y,k)$.

\section{List of octahedron relations}\label{sec:octlist}
In this section we give a list of the polynomials $\Omega_{1}$ and $\Omega_{2}$ except in the cases coming from $Q_{4}$ and $A_{2}$. In those cases we give the polynomials $R_{1}$. The polynomials $R_{2}$ and $R_{3}$ can be obtained from $R_{1}$ by cyclic permutations $(123)$.
\begin{itemize}
\item[$Q_{1}^{\delta}$:] $\Omega_{1}=\alpha_{1}(x_{2}x_{12}-x_{3}x_{13})+\alpha_{2}(x_{3}x_{23}-x_{1}x_{12})+
    \alpha_{3}(x_{1}x_{13}-x_{2}x_{23})$\\
$\phantom{\Omega_{1}=}-\alpha_{1}x_{1}(x_{2}-x_{3})-\alpha_{2}x_{2}(x_{3}-x_{1})-\alpha_{3}x_{3}(x_{1}-x_{2})$\\
$\phantom{\Omega_{1}=}+\alpha_{1}x_{23}(x_{13}-x_{12})+\alpha_{2}x_{13}(x_{12}-x_{23})+
\alpha_{3}x_{12}(x_{23}-x_{13})$\par
$\Omega_{2}=x_{1}x_{2}(x_{23}-x_{13})+x_{2}x_{3}(x_{13}-x_{12})+x_{3}x_{1}(x_{12}-x_{23})$\\
$\phantom{\Omega_{2}=}+x_{23}x_{13}(x_{1}-x_{2})+x_{13}x_{12}(x_{2}-x_{3})+x_{12}x_{23}(x_{3}-x_{1})$\\
$\phantom{\Omega_{2}=}-\delta(\alpha_{1}(\alpha_{2}-\alpha_{3})(x_{1}+x_{23})+
\alpha_{2}(\alpha_{3}-\alpha_{1})(x_{2}+x_{13})+\alpha_{3}(\alpha_{1}-\alpha_{2})(x_{3}+x_{12}))$
\item[$Q_{2}$:]  $\Omega_{1}=\alpha_{1}(x_{2}x_{12}-x_{3}x_{13})+\alpha_{2}(x_{3}x_{23}-x_{1}x_{12})+
    \alpha_{3}(x_{1}x_{13}-x_{2}x_{23})$\\
$\phantom{\Omega_{1}=}-\alpha_{1}x_{1}(x_{2}-x_{3})-\alpha_{2}x_{2}(x_{3}-x_{1})-\alpha_{3}x_{3}(x_{1}-x_{2})$\\
$\phantom{\Omega_{1}=}+\alpha_{1}x_{23}(x_{13}-x_{12})+\alpha_{2}x_{13}(x_{12}-x_{23})+\alpha_{3}x_{12}(x_{23}-x_{13})$\\
$\phantom{\Omega_{1}=}+\alpha_{1}^{2}(\alpha_{2}-\alpha_{3})(x_{1}+x_{2}-x_{23}-x_{13})+
\alpha_{2}^{2}(\alpha_{3}-\alpha_{1})(x_{2}+x_{3}-x_{13}-x_{12})$\\
$\phantom{\Omega_{1}=}+\alpha_{3}^{2}(\alpha_{1}-\alpha_{2})(x_{3}+x_{1}-x_{12}-x_{23})$\par
$\Omega_{2}=x_{1}x_{2}(x_{23}-x_{13})+x_{2}x_{3}(x_{13}-x_{12})+x_{3}x_{1}(x_{12}-x_{23})$\\
$\phantom{\Omega_{2}=}+x_{23}x_{13}(x_{1}-x_{2})+x_{13}x_{12}(x_{2}-x_{3})+x_{12}x_{23}(x_{3}-x_{1})$\\
$\phantom{\Omega_{2}=}-2\alpha_{1}(\alpha_{2}-\alpha_{3})x_{1}x_{23}-2\alpha_{2}(\alpha_{3}-\alpha_{1})x_{2}x_{13}-
2\alpha_{3}(\alpha_{1}-\alpha_{2})x_{3}x_{12}$\\
$\phantom{\Omega_{2}=}+\alpha_{1}(\alpha_{2}-\alpha_{3})(x_{2}x_{3}+x_{13}x_{12})+\alpha_{2}(\alpha_{3}-
\alpha_{1})(x_{3}x_{1}+x_{12}x_{23})$\\
$\phantom{\Omega_{2}=}+\alpha_{3}(\alpha_{1}-\alpha_{2})(x_{1}x_{2}+x_{23}x_{13})+\alpha_{1}(\alpha_{2}-
\alpha_{3})(x_{2}x_{12}+x_{3}x_{13})$\\
$\phantom{\Omega_{2}=}+\alpha_{2}(\alpha_{3}-\alpha_{1})(x_{3}x_{23}+x_{1}x_{12})+\alpha_{3}(\alpha_{1}-
\alpha_{2})(x_{1}x_{13}+x_{2}x_{23})$\\
$\phantom{\Omega_{2}=}+\alpha_{1}(\alpha_{2}-\alpha_{3})\left(\alpha_{1}^{2}+\alpha_{2}^{2}+\alpha_{3}^{3}-
\alpha_{1}(\alpha_{2}+\alpha_{3})\right)(x_{1}+x_{23})$\\
$\phantom{\Omega_{2}=}+\alpha_{2}(\alpha_{3}-\alpha_{1})\left(\alpha_{2}^{2}+\alpha_{3}^{2}+\alpha_{1}^{2}-
\alpha_{2}(\alpha_{3}+\alpha_{1})\right)(x_{2}+x_{13})$\\
$\phantom{\Omega_{2}=}+\alpha_{3}(\alpha_{1}-\alpha_{2})\left(\alpha_{3}^{2}+\alpha_{1}^{2}+\alpha_{2}^{2}-
\alpha_{3}(\alpha_{1}+\alpha_{2})\right)(x_{3}+x_{12})$\\
$\phantom{\Omega_{2}=}+2\alpha_{1}\alpha_{2}\alpha_{3}(\alpha_{1}-\alpha_{2})(\alpha_{2}-\alpha_{3})
(\alpha_{3}-\alpha_{1})$
\item[$Q_{3}^{\delta}$:] $\Omega_{1}=\sin(\alpha_{1})(x_{2}x_{12}-x_{3}x_{13})+\sin(\alpha_{2})(x_{3}x_{23}-x_{1}x_{12})+
    \sin(\alpha_{3})(x_{1}x_{13}-x_{2}x_{23})$\\
$\phantom{\Omega_{1}=}-\sin(\alpha_{1}-\alpha_{2})(x_{1}x_{2}-x_{23}x_{13})-
\sin(\alpha_{2}-\alpha_{3})(x_{2}x_{3}-x_{13}x_{12})$\\
$\phantom{\Omega_{1}=}-\sin(\alpha_{3}-\alpha_{1})(x_{3}x_{1}-x_{12}x_{23})$\par
$\Omega_{2}=\cos({\scriptstyle\frac{1}{2}}(\alpha_{1}-\alpha_{2}-\alpha_{3}))(x_{1}+x_{23})(x_{2}x_{13}-x_{3}x_{12})$\\
$\phantom{\Omega_{2}=}+\cos({\scriptstyle\frac{1}{2}}(\alpha_{2}-\alpha_{3}-\alpha_{1}))(x_{2}+x_{13})
(x_{3}x_{12}-x_{1}x_{23})$\\
$\phantom{\Omega_{2}=}+\cos({\scriptstyle\frac{1}{2}}(\alpha_{3}-\alpha_{1}-\alpha_{2}))(x_{3}+x_{12})
(x_{1}x_{23}-x_{2}x_{13})$\\
$\phantom{\Omega_{2}=}+4\delta\sin(\alpha_{1})\sin(\alpha_{2}-\alpha_{3})\cos({\scriptstyle\frac{1}{2}}
(\alpha_{1}-\alpha_{2}-\alpha_{3}))(x_{1}+x_{23})$\\
$\phantom{\Omega_{2}=}+4\delta\sin(\alpha_{2})\sin(\alpha_{3}-\alpha_{1})\cos({\scriptstyle\frac{1}{2}}
(\alpha_{2}-\alpha_{3}-\alpha_{1}))(x_{2}+x_{13})$\\
$\phantom{\Omega_{2}=}+4\delta\sin(\alpha_{3})\sin(\alpha_{1}-\alpha_{2})\cos({\scriptstyle\frac{1}{2}}
(\alpha_{3}-\alpha_{1}-\alpha_{2}))(x_{3}+x_{12})$
\item[$Q_{4}$:] $R_{1}=\sn(\alpha_{1})\tilde{R}_{1}$ with $\tilde{R}_{1}=\sn(\alpha_{3})(x_{1}x_{13}-x_{2}x_{23})+\sn(\alpha_{2})(x_{3}x_{23}-x_{1}x_{12})$\\
$\phantom{\Omega_{1}=}+\sn(\alpha_{1})\left(1+k^{2}\sn(\alpha_{3})\sn(\alpha_{2})\sn(\alpha_{3}-\alpha_{1})
\sn(\alpha_{1}-\alpha_{2})\right)(x_{2}x_{12}-x_{3}x_{13})$\\
$\phantom{\Omega_{1}=}-\sn(\alpha_{3}-\alpha_{1})(x_{3}x_{1}-x_{12}x_{23})-\sn(\alpha_{1}-\alpha_{2})
(x_{1}x_{2}-x_{23}x_{13})$\\
$\phantom{\Omega_{1}=}-\sn(\alpha_{2}-\alpha_{3})\left(1+k^{2}\sn(\alpha_{3})\sn(\alpha_{2})\sn(\alpha_{3}-\alpha_{1})
\sn(\alpha_{1}-\alpha_{2})\right)(x_{2}x_{3}-x_{13}x_{12})$\\
$\phantom{\Omega_{1}=}+k^{2}\sn(\alpha_{3})\sn(\alpha_{2})\sn(\alpha_{1}-\alpha_{2})
(x_{1}x_{2}x_{13}x_{12}-x_{2}x_{3}x_{23}x_{13})$\\
$\phantom{\Omega_{1}=}-k^{2}\sn(\alpha_{3})\sn(\alpha_{2})\sn(\alpha_{3}-\alpha_{1})
(x_{2}x_{3}x_{12}x_{23}-x_{3}x_{1}x_{13}x_{12})$\\
$\phantom{\Omega_{1}=}+k^{2}\sn(\alpha_{3})\sn(\alpha_{3}-\alpha_{1})\sn(\alpha_{1}-\alpha_{2})
(x_{3}x_{1}x_{2}x_{13}-x_{2}x_{13}x_{12}x_{23})$\\
$\phantom{\Omega_{1}=}-k^{2}\sn(\alpha_{2})\sn(\alpha_{3}-\alpha_{1})\sn(\alpha_{1}-\alpha_{2})
(x_{1}x_{2}x_{3}x_{12}-x_{3}x_{12}x_{23}x_{13})$
\item[$H_{1}$:] $\Omega_{1}=x_{1}(x_{13}-x_{12})+x_{2}(x_{12}-x_{23})+x_{3}(x_{23}-x_{13})$\par
$\Omega_{2}=x_{1}x_{2}(x_{23}-x_{13})+x_{2}x_{3}(x_{13}-x_{12})+x_{3}x_{1}(x_{12}-x_{23})$\\
$\phantom{\Omega_{2}=}+x_{23}x_{13}(x_{1}-x_{2})+x_{13}x_{12}(x_{2}-x_{3})+x_{12}x_{23}(x_{3}-x_{1})$\\
$\phantom{\Omega_{2}=}+(\alpha_{1}-\alpha_{2})(x_{3}+x_{12})+(\alpha_{2}-\alpha_{3})(x_{1}+x_{23})+
(\alpha_{3}-\alpha_{1})(x_{2}+x_{13})$
\item[$H_{2}$:] $\Omega_{1}=x_{1}(x_{13}-x_{12})+x_{2}(x_{12}-x_{23})+x_{3}(x_{23}-x_{13})$\\
$\phantom{\Omega_{1}=}-\alpha_{1}(x_{2}-x_{3}+x_{13}-x_{12})-\alpha_{2}(x_{3}-x_{1}+x_{12}-x_{23})-
\alpha_{3}(x_{1}-x_{2}+x_{23}-x_{13})$\par
$\Omega_{2}=x_{1}x_{2}(x_{23}-x_{13})+x_{2}x_{3}(x_{13}-x_{12})+x_{3}x_{1}(x_{12}-x_{23})$\\
$\phantom{\Omega_{2}=}+x_{23}x_{13}(x_{1}-x_{2})+x_{13}x_{12}(x_{2}-x_{3})+x_{12}x_{23}(x_{3}-x_{1})$\\
$\phantom{\Omega_{2}=}+2(\alpha_{1}-\alpha_{2})x_{3}x_{12}+2(\alpha_{2}-\alpha_{3})x_{1}x_{23}+
2(\alpha_{3}-\alpha_{1})x_{2}x_{13}$\\
$\phantom{\Omega_{2}=}+(\alpha_{1}-\alpha_{2})(x_{1}x_{2}+x_{23}x_{13})+
(\alpha_{2}-\alpha_{3})(x_{2}x_{3}+x_{13}x_{12})$\\
$\phantom{\Omega_{2}=}+(\alpha_{3}-\alpha_{1})(x_{3}x_{1}+x_{12}x_{23})$\\
$\phantom{\Omega_{2}=}+\left(\alpha_{1}^{2}-\alpha_{2}^{2}\right)(x_{3}+x_{12})+
\left(\alpha_{2}^{2}-\alpha_{3}^{2}\right)(x_{1}+x_{23})+\left(\alpha_{3}^{2}-\alpha_{1}^{2}\right)(x_{2}+x_{13})$\\
$\phantom{\Omega_{2}=}-2(\alpha_{1}-\alpha_{2})(\alpha_{2}-\alpha_{3})(\alpha_{3}-\alpha_{1})$
\item[$H_{3}^{\delta}$:] $\Omega_{1}=e^{\alpha_{3}}x_{1}(x_{13}-x_{12})+e^{\alpha_{1}}x_{2}(x_{12}-x_{23})+
    e^{\alpha_{2}}x_{3}(x_{23}-x_{13})$\\
$\Omega_{2}=x_{1}x_{2}\left(e^{\alpha_{3}+\alpha_{1}}x_{23}-e^{\alpha_{2}+\alpha_{3}}x_{13}\right)+
x_{2}x_{3}\left(e^{\alpha_{1}+\alpha_{2}}x_{13}-e^{\alpha_{3}+\alpha_{1}}x_{12}\right)$\\
$\phantom{\Omega_{2}=}+x_{3}x_{1}\left(e^{\alpha_{2}+\alpha_{3}}x_{12}-e^{\alpha_{1}+\alpha_{2}}x_{23}\right)+
x_{23}x_{13}\left(e^{\alpha_{3}+\alpha_{1}}x_{1}-e^{\alpha_{2}+\alpha_{3}}x_{2}\right)$\\
$\phantom{\Omega_{2}=}+x_{13}x_{12}\left(e^{\alpha_{1}+\alpha_{2}}x_{2}-e^{\alpha_{3}+\alpha_{1}}x_{3}\right)+
x_{12}x_{23}\left(e^{\alpha_{2}+\alpha_{3}}x_{3}-e^{\alpha_{1}+\alpha_{2}}x_{1}\right)$\\
$\phantom{\Omega_{2}=}-\delta\left(e^{\alpha_{3}+2\alpha_{1}}-e^{\alpha_{3}+2\alpha_{2}}\right)
\left(x_{3}+x_{12}\right)-\delta\left(e^{\alpha_{1}+2\alpha_{2}}-e^{\alpha_{1}+2\alpha_{3}}\right)(x_{1}+x_{23})$\\
$\phantom{\Omega_{2}=}-\delta\left(e^{\alpha_{2}+2\alpha_{3}}-e^{\alpha_{2}+2\alpha_{1}}\right)(x_{2}+x_{13})$
\item[$A_{1}^{\delta}$:] $\Omega_{1}=\alpha_{1}(x_{2}x_{12}-x_{3}x_{13})+\alpha_{2}(x_{3}x_{23}-x_{1}x_{12})+
    \alpha_{3}(x_{1}x_{13}-x_{2}x_{23})$\\
$\phantom{\Omega_{1}=}+\alpha_{1}x_{1}(x_{2}-x_{3})+\alpha_{2}x_{2}(x_{3}-x_{1})+\alpha_{3}x_{3}(x_{1}-x_{2})$\\
$\phantom{\Omega_{1}=}-\alpha_{1}x_{23}(x_{13}-x_{12})-\alpha_{2}x_{13}(x_{12}-x_{23})-
\alpha_{3}x_{12}(x_{23}-x_{13})$\\
$\Omega_{2}=x_{1}x_{2}((\alpha_{1}-\alpha_{2}+\alpha_{3})x_{23}-(\alpha_{3}-\alpha_{1}+\alpha_{2})x_{13})$\\
$\phantom{\Omega_{2}=}+x_{2}x_{3}((\alpha_{2}-\alpha_{3}+\alpha_{1})x_{13}-
(\alpha_{1}-\alpha_{2}+\alpha_{3})x_{12})$\\
$\phantom{\Omega_{2}=}+x_{3}x_{1}((\alpha_{3}-\alpha_{1}+\alpha_{2})x_{12}-
(\alpha_{2}-\alpha_{3}+\alpha_{1})x_{23})$\\
$\phantom{\Omega_{2}=}+x_{23}x_{13}((\alpha_{1}-\alpha_{2}+\alpha_{3})x_{1}-
(\alpha_{3}-\alpha_{1}+\alpha_{2})x_{2})$\\
$\phantom{\Omega_{2}=}+x_{13}x_{12}((\alpha_{2}-\alpha_{3}+\alpha_{1})x_{2}-
(\alpha_{1}-\alpha_{2}+\alpha_{3})x_{3})$\\
$\phantom{\Omega_{2}=}+x_{12}x_{23}((\alpha_{3}-\alpha_{1}+\alpha_{2})x_{3}-
(\alpha_{2}-\alpha_{3}+\alpha_{1})x_{1})$\\
$\phantom{\Omega_{2}=}-\delta\alpha_{1}(\alpha_{2}-\alpha_{3})(\alpha_{1}-\alpha_{2}-\alpha_{3})(x_{1}+x_{23})$\\
$\phantom{\Omega_{2}=}-\delta\alpha_{2}(\alpha_{3}-\alpha_{1})(\alpha_{2}-\alpha_{3}-\alpha_{1})(x_{2}+x_{13})$\\
$\phantom{\Omega_{2}=}-\delta\alpha_{3}(\alpha_{1}-\alpha_{2})(\alpha_{3}-\alpha_{1}-\alpha_{2})(x_{3}+x_{12})$
\item[$A_{2}$:] $R_{1}=\sin(\alpha_{3})\sin(\alpha_{1})(x_{3}x_{23}-x_{1}x_{12})+
    \sin(\alpha_{1})\sin(\alpha_{2})(x_{1}x_{13}-x_{2}x_{23})$\\
$\phantom{\Omega_{1}=}+\sin(\alpha_{2})\sin(\alpha_{3})(x_{2}x_{12}-x_{3}x_{13})+
\sin(\alpha_{3}-\alpha_{1})\sin(\alpha_{1}-\alpha_{2})(x_{3}x_{13}-x_{2}x_{12})$\\
$\phantom{\Omega_{1}=}+\sin(\alpha_{1})\sin(\alpha_{3}-\alpha_{1})(x_{3}x_{1}x_{13}x_{12}-x_{2}x_{3}x_{12}x_{23})$\\
$\phantom{\Omega_{1}=}+\sin(\alpha_{1})\sin(\alpha_{1}-\alpha_{2})(x_{1}x_{2}x_{13}x_{12}-x_{2}x_{3}x_{23}x_{13})$
\end{itemize}
\small
\bibliographystyle{amsalpha}
\bibliography{Quellen}
\end{document}